\documentclass[10pt,conference]{IEEEtran} 

\usepackage{times}
\usepackage[latin1]{inputenc}
\usepackage[T1]{fontenc}
\usepackage[francais,german,greek,english]{babel}
\usepackage{epsf}
\usepackage{epsfig}
\usepackage{psfrag}
\usepackage{amsbsy}
\usepackage{amsmath}
\usepackage{amsthm}
\usepackage{fixmath}  
\usepackage{graphicx}
\usepackage{subfigure}
\usepackage{amssymb}
\usepackage{latexsym}
\usepackage[usenames]{color}
\usepackage{rotating}
\usepackage{pstcol}
\usepackage{algorithmic}
\usepackage{algorithm}
\usepackage{tabularx}
\usepackage{cite}
\usepackage{dblfloatfix}
\usepackage{auto-pst-pdf}

\newcommand{\BS}{\boldsymbol}

\DeclareMathOperator*{\Cov}{Cov}

\DeclareMathOperator*{\Var}{Var}
\DeclareMathOperator*{\E}{E}

\begin{document}

\title{Timing Channel: Achievable Rate in the Finite Block-Length Regime}

\author{
\IEEEauthorblockN{Thomas~J.~Riedl, Todd P. Coleman and Andrew~C.~Singer} 
\IEEEauthorblockA{University of Illinois at Urbana-Champaign\\ 
Urbana, IL 61801, USA \\ 
Email: triedl2@illinois.edu}
\thanks{This work was supported in part by the department of the Navy, Office of 
Naval Research, under grants ONR MURI N00014-07-1-0738 and ONR
N00014-07-1-0311.}
}

\maketitle

\section{Introduction}

\section{Basic Definitions and Conventions}

\begin{itemize}
\item Symbols denoting random variables are capitalized but symbols denoting their realization are not.
\item Bold face is used for vectors.
\item Blackboard bold or non-italic face are used for sets, calligraphic face for system of sets.
\item We use $\bar{x}$ as a shorthand for the difference $1-x$ for $x \in [0,1]$.
\item $\mathbb{Z}$ is the set of all integers and $\mathbb{Z}_{+} = \{z \in \mathbb{Z}: z \geq 0  \}$.
\item $\BS{2}$ is the set $\{0,1\}$.
\item $[n]$ denotes the set $\{z \in \mathbb{Z}: 1 \leq z \leq n  \}$.
\item Given the input and output sets $\mathrm{X}^n$ and $\mathrm{Y}^n$ a channel is a function $P_{Y^n|X^n}(\cdot | \cdot): \mathcal{A}  \times \mathrm{X}^n \to [0,1]$ such that for each $\BS{x}^n \in \mathrm{X}^n$  $P_{Y^n|X^n}( \cdot | \BS{x}^n)$ is a probability measure on $\mathcal{A}$. $\mathcal{A}$ is some $\sigma$-algebra over $\mathrm{Y}^n$.
\item We define the finite sets $\mathrm{X}^n$ and $\mathrm{Y}^n$ to be $\mathbb{Z}_{+} \times \BS{2}^{n-1} $ and $\BS{2}^{n}$, repectively, and associate the $\sigma$-algebras $\mathcal{P}(\mathrm{X}^n)$ and $\mathcal{P}(\mathrm{Y}^n)$ with them. $\mathcal{P}$ indicates the power set.
\item Given a distribution $P_{X^n}$ on $\mathrm{X}^n$ we define the distribution $P_{Y^n} = \int_{\mathrm{X}^n} P_{Y^n|X^n}(\BS{y}^n|\BS{x}^n) P_{X^n}(d X^n)$ and the function $i(\BS{x}^n,\BS{y}^n) = \log \frac{P_{Y^n|X^n}(\BS{y}^n|\BS{x}^n)}{P_{Y^n}(\BS{y}^n)}$ for $\BS{x}^n \in \mathrm{X}^n$ and $\BS{y}^n \in \mathrm{Y}^n$. This function is called the information density. 
\item The expected value of a random variable $X$ is denoted by $\mathbb{E}[X]$.
\item The mutual information between two random vectors $\BS{X}^n$ and $\BS{Y}^n$ is $\mathbb{E}[i(\BS{x}^n,\BS{y}^n)]$.
\item The capacity $C(\lambda)$ of the channel is $\sup_{P_{X^n}} \mathbb{E}[i(\BS{x}^n,\BS{y}^n)]$ where $P_{X^n}$ is such that the rate $\mathbb{E}[\sum_{i=1}^{n-1} X_i]/(n-1)$ does not exceed $\lambda$.
\item The binary entropy function $H(p)= - p \log p - \bar{p} \log \bar{p}$.
\item We drop subscripts whenever they are clear from the context. For example $P_{Y^n|X^n}(\BS{y}^n|\BS{x}^n) = P(\BS{y}^n|\BS{x}^n)$.
\end{itemize}

\section{System Description and Preliminaries}

The communication channel we consider is an interesting example of a channel with memory. It is essentially a probabilistic single server queuing system with the length of the queue being the memory of channel. At each discrete time instance $i$ the random variable $\tilde{X}_i$, $i \in [n-1]$ indicates if there was an arrival at the back of the queue at time $i$. The initial length of the queue $Q_0$ is a non-negative integer-valued random variable with distribution $P_{Q_0}$. The variable $Q_0$ together with the vector $\tilde{X}_i$, $i \in [n-1]$ is considered the channel input vector ${\BS{X}}^n \in \mathrm{X}^n$.
We define the random variables $X_i, \tilde{Y}_i, Y_i$ and  $Q_i$ for $i \in  [n-1] \cup \{ 0 \}$ such that
\begin{align}
X_i &= \sum_{l = 1}^i \tilde{X}_l \\
Y_i &= \sum_{l = 0}^i \tilde{Y}_l \\
Q_i &= Q_0 + X_i - Y_{i-1} = Q_{i-1} + \tilde{X}_i - \tilde{Y}_{i-1}
\end{align}
where the binary random variables $\tilde{Y}_{i}$ are conditionally independent given $Q_i$ and are distributed according to the transition probability function
\begin{align}
P_{\tilde{Y}_{i}|Q_i}(\tilde{Y}_{i}|Q_i) = 
\left\{ \begin{array}{rl} 
1; &\tilde{Y}_{i} = 0, Q_i = 0 \\ 
\bar{\mu}; &\tilde{Y}_{i} = 0, Q_i > 0  \\ 
\mu; &\tilde{Y}_{i} = 1, Q_i > 0
\end{array} \right.
\end{align}
The vector $\tilde{Y}_i$, $i \in [n-1] \cup \{ 0 \}$ is considered the channel output vector $\BS{Y}^n \in \mathrm{Y}^n$.

Clearly, $X_i$ ($Y_i$) counts the total number of arrivals (departures), $Q_i$ denotes the length of the queue at time $i$ and $\tilde{Y}_{i}$ indicates if there was a departure from the front of the queue at time $i$. The described relationships are illustrated in Figure \ref{timingChannel}.
\begin{figure}[!hbtp]
\psfrag{xt}{$\tilde{X}_i$}
\psfrag{x}{$X_i$}
\psfrag{q}{$Q_i$}
\psfrag{Z}{$z^{-1}$}
\psfrag{E}{$\sum$}
\psfrag{y}{$Y_i$}
\psfrag{yt}{$\tilde{Y}_{i}$}
\psfrag{e}{$\bar{\mu}$}
\psfrag{Q=0}{$\scriptstyle{Q=0}$}
\psfrag{Q>0}{$\scriptstyle{Q>0}$}
\psfrag{0}{$\scriptstyle{0}$}
\psfrag{1}{$\scriptstyle{1}$}
\centering
\includegraphics[width = 0.9\columnwidth]{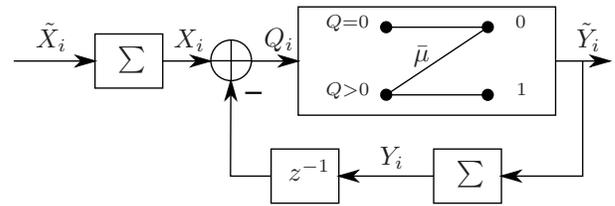}
\caption{Timing Channel Diagram}
\label{timingChannel}
\end{figure}

In principle the distribution on $\BS{X}^n$ could be chosen arbitrary but we will assume $\tilde{X}_i$, $i \in [n-1]$ to be i.i.d. Bernoulli$(\lambda)$ random variables for the following reason:
\newtheorem{theorem}{Theorem}
\begin{theorem}
The distribution on $\BS{X}^n$ that achieves capacity makes $\tilde{X}_i$, $i \in [n-1]$ i.i.d. Bernoulli$(\lambda)$. 
\end{theorem}
\begin{proof}
The proof can be found in \cite{Bedekar98onthe}.
\end{proof}

Most of the analysis in this paper is built on the theory of Markov chains and so clearly we have to start with a proper definition of this kind of stochastic process. Given a probability space $(\Omega,\mathcal{F},P)$ a discrete time stochastic process $\Phi$ on a state space $\mathbb{X}$ is, for our purposes, a collection of random variables $\Phi_i$, $i \in \mathbb{N}_0$ where each $\Phi_i$ takes values in the set $\mathbb{X}$. The random variables are measurable with respect to some given $\sigma$-field $\mathcal{A}$ over $\mathbb{X}$. We only consider countable state spaces here and can hence take $\Omega$ to be the product space $\prod_{i = 0}^\infty \mathbb{N}_0$ and take $\mathcal{F}$ to be the product $\sigma$-algebra $\bigvee_{i=0}^\infty \sigma(\mathbb{N}_0)$ defined in the usual way. If $P(\Phi_{i+1} = \phi_{i+1}|\Phi_i = \phi_i, \ldots, \Phi_0 = \phi_0) = P(\Phi_{i+1} = \phi_{i+1}|\Phi_i = \phi_i)$ holds, we call the process Markov. Given an initial measure $P_{\Phi_0}$ and transition probabilities $P_{\Phi_{i+1}|\Phi_i}(\phi_{i+1}|\phi_i)$  there exists a probability law $P$ satisfying $P(\Phi_{i+1} = \phi_{i+1},\Phi_i = \phi_i, \ldots, \Phi_0 = \phi_0) =  P_{\Phi_0}(\phi_0) \prod_{l=0}^{i} P_{\Phi_{l+1}|\Phi_l}(\phi_{l+1}|\phi_l)$ by the Kolmogorov extension theorem \cite{meyntwee09,gray09}.

The length of the queue $Q_i$ forms an irreducible Markov chain $Q = \{\Omega,\mathcal{F},Q_i,P_Q \}$ with transition probabilities
\begin{align}
P_{Q_{i+1}|Q_i}(q_{i+1}|q_i) = 
\left\{ \begin{array}{rl} 
\lambda; & q_{i+1}=q_i+1,q_i = 0 \\ 
\bar{\lambda}; & q_{i+1} = q_i, q_i = 0 \\ 
\bar{\lambda}\mu; & q_{i+1} = q_i-1 , q_i > 0 \\ 
{\lambda}{\bar{\mu}}; & q_{i+1} =  q_i+1 , q_i > 0 \\ 
1-{\lambda}{\bar{\mu}}-\bar{\lambda}\mu; & q_{i+1} = q_i , q_i > 0 \\ 
\end{array} \right.
\end{align}
If and only if $\lambda < \mu$, there exists a probability measure $\pi_{Q}$ on $\mathbb{N}_0$ that solves the system of equations 
\begin{align}
\sum_{q_i \in  \mathbb{N}_0} \pi_{Q}(q_i) P_{Q_{i+1}|Q_i}(q_{i+1}|q_i) = \pi_{Q}(q_{i+1})
\end{align}
for all $q_{i+1} \in \mathbb{N}_0$ and this measure is called the invariant measure. Note that for irreducible Markov chains the existence of such a probability measure is equivalent to positive recurrence. For the transition probabilities given it can be checked that
\begin{align}
 \pi_{Q}(q_i) = 
\left\{ \begin{array}{rl} 
\frac{\bar{\lambda}\mu-\lambda\bar{\mu}}{\mu}; q_i = 0 \\ 
\frac{\bar{\lambda}\mu-\lambda\bar{\mu}}{\bar{\mu} \mu} \rho^{q_i}; q_i > 0
\end{array} \right.
\end{align}
where we defined
\begin{align}
\rho = \frac{\lambda \bar{\mu}}{\bar{\lambda} \mu}.
\end{align}
Note that $\rho < 1 $ if and only if $\lambda < \mu$. In the remainder of the paper we will always assume that the arrival rate $\lambda$ is smaller than the serving rate $\mu$.
\begin{theorem}
The distribution on $\BS{X}^n$ that achieves capacity makes $Q_0$ distributed according to $\pi_{Q}$.
\end{theorem}
\begin{proof}
The proof can be found in \cite{Bedekar98onthe}.
\end{proof}
And we will hence always assume that $Q_0$ is distributed according to $\pi_{Q}$. 

\newtheorem{definition}{Definition}
\begin{definition}
Let $0<\epsilon<1$. An $\epsilon$-Code of size $N$ is defined as a sequence $\{ (\BS{x}^{(i)}, \BS{D}^{(i)}), i = 1,\dots,N \} $ such that $\BS{x}^{(i)} \in \mathbb{X}^n$, $\BS{D}^{(i)} \subset  \mathbb{Y}^n$,  $\BS{D}^{(i)}$ are mutually disjoint and $P(\BS{D}^{(i)}|\BS{x}^{(i)}) > 1-\epsilon \ \forall i$. Let $N(\epsilon,n)$ be the supremum of the set of integers $N$ such that an $\epsilon$-Code of size $N$ exists. 
\end{definition}
Note that $\frac{\log N(\epsilon,n)}{n}$ equals the rate of the code.

\begin{theorem}
The expression for the capacity $C$ of the channel simplifies to
\begin{align}
C &= H(\lambda) - \frac{\lambda}{\mu} H(\mu)
\end{align}
and
\begin{align}
\log N(\epsilon,n) = nC + \epsilon\mathcal{O}(n)
\end{align}
\end{theorem}
\begin{proof}
A proof of a similar result for the continuos  time case appeared in the landmark paper ``Bits through Queues'' \cite{Anantharam96} and the stated discrete time result was proved in \cite{Bedekar98onthe}.
\end{proof}

\begin{theorem}
\label{theorem:burke}
\emph{(Burke's Theorem)}
Given the queue $Q$ is in equilibrium and the random variables $\tilde{X}_i$, $i \in [n-1]$ are i.i.d. Bernoulli$(\lambda)$ then the random variables $\tilde{Y}_i$, $i \in [n-1] \cup \{ 0 \}$ are also i.i.d. Bernoulli$(\lambda)$.
\end{theorem}
\begin{proof}
The proof is similar to the one for continuous time queues and can be found in \cite{takagi93}. A concise proof is also given in the appendix.
\end{proof}
By the above theorem $P(\BS{y}^n) = \prod_{i=0}^n P(\tilde{y}_i)$ and
\begin{align}
P(\tilde{y}_i) = \lambda \cdot  \tilde{y}_i + \bar{\lambda} \cdot \bar{\tilde{y}}_i.
\end{align}

The following lemma uses arguments introduced by Feinstein \cite{feinstein54} to give a lower bound on $N(\epsilon,n)$.
\newtheorem{lemma}{Lemma}
\begin{lemma}
\label{lemma:feinstein}
\emph{(Feinstein)}
$N(\epsilon,n) \geq e^{\theta} \left( \epsilon - P(i(\BS{x}^n,\BS{y}^n) \leq \theta) \right)$ for all $\theta \in \mathbb{R}$.
\end{lemma}
\begin{proof}
The proof is short and elegant and reproduced in the appendix.
\end{proof}

\section{Finite-Length Scaling}

The distributions $P(\BS{y}^n|\BS{x}^n)$ and $P(\BS{y}^n)$ factor and hence
\begin{align}
i(\BS{x},\BS{y}) = \sum_{i=0}^{n-1} \log \frac{P(\tilde{y}_i|q_i)}{P(\tilde{y}_i)} =  \sum_{i=0}^{n-1} f(\tilde{y}_i,q_i)
\end{align}
where we defined
\begin{align}
f(\tilde{y}_i,q_i) = \log \frac{P(\tilde{y}_i|q_i)}{P(\tilde{y}_i)}.
\end{align}

The composed state $\psi_i = (\tilde{y}_i,q_i)$ again forms a positive recurrent Markov chain $\Psi = \{\Omega,\mathcal{F},\Psi_i,P_\Psi \}$ whose transition probabilities are illustrated in Figure \ref{2DMarkovModel}.
\begin{figure}[!hbtp]
\psfrag{Q=0,Yt=0}{$q=0,\tilde{y}=0$}
\psfrag{Q=1,Yt=0}{$q=1,\tilde{y}=0$}
\psfrag{Q=2,Yt=0}{$q=2,\tilde{y}=0$}
\psfrag{Q=1,Yt=1}{$q=1,\tilde{y}=1$}
\psfrag{Q=2,Yt=1}{$q=2,\tilde{y}=1$}
\psfrag{ldmd}{$\bar{\lambda}\bar{\mu}$}
\psfrag{lmd}{${\lambda}\bar{\mu}$}
\psfrag{ld}{$\bar{\lambda}$}
\psfrag{lm}{${\lambda}\mu$}
\psfrag{ldm}{$\bar{\lambda}\mu$}
\centering
\includegraphics[width = 0.98\columnwidth]{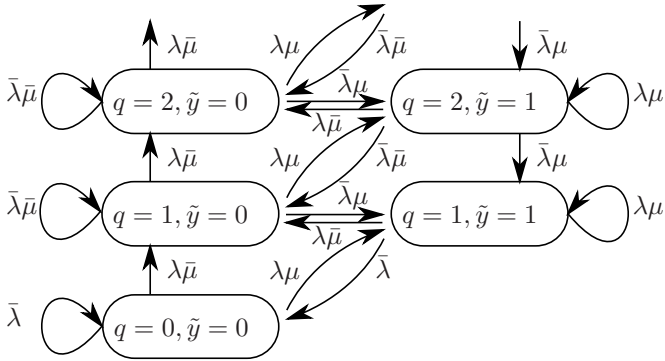}
\caption{Possible Transitions in the Markov Chain $(\tilde{Y},Q)$}
\label{2DMarkovModel}
\end{figure}
The invariant measure $\pi_\Psi$ for this chain is only a slight extension to $\pi_Q$:
\begin{align}
\pi_\Psi(\tilde{y},q) = P_{\tilde{Y}_{i}|Q_i}(\tilde{y}|q) \pi_Q(q) 
\end{align}
The proof of the following theorem is one of the main contributions of this paper because it can be used to proof bounds and an asymptotic on the quantity $N(\epsilon,n)$.
\begin{theorem}
\label{theorem:BerryEsseen4MCs}
The asymptotic variance 
\begin{align}
\label{eq:asymVar}
\sigma^2 = \lim_{n \to \infty} \frac{1}{n} \Var(\sum_{i=0}^{n-1} f(\tilde{y}_i,q_i))
\end{align}
is well defined, positive and finite, and
\begin{align}
\label{eq:asymVarianceRep}
\sigma^2 = \Var(f(\Psi_0)) + 2 \sum_{i=1}^\infty \Cov(f(\Psi_0),f(\Psi_i)).
\end{align}
Further the following Berry-Esseen type bound holds:
\begin{align}
\label{eq:BerryEsseen4MCs}
\sup_{\xi \in \mathbb{R}} \left| P \left(\frac{\sum_{i=0}^{n-1} f(\tilde{y}_i,q_i) - n \pi_\Psi(f)}{\sigma \sqrt{n}} \leq \xi \right) - \Phi(\xi) \right| \leq \mathcal{O}(n^{-1/2})
\end{align}
\end{theorem}
\begin{proof}
A detailed proof can be found in appendix. We only give a sketch here. The Markov Chain $\Psi$ is aperiodic and irreducible. The state space of $\Psi_i$ can be chosen to be $\mathbb{X} = \BS{2} \times \mathbb{N} \cup \{ (0,0) \}$. First we verify that there exists a Lyapunov function $V: \mathbb{X} \to (0,\infty]$, finite at some $\psi_0 \in \mathbb{X}$, a finite set $\mathbb{S} \subset \mathbb{X}$, and $b < \infty$ such that
\begin{align}
\E [ V(\Psi_{i+1}) - V(\Psi_{i})|\Psi_{i} = \psi ] \leq -1 + b \BS{1}_{\mathbb{S}}(\psi), \quad \psi \in \mathbb{X}
\end{align}
The chain is skip-free and the found Lyapunov function is linear and hence also Lipschitz. These properties imply that the chain is geometric ergodic \cite{spieksma1994,CTCNMeyn2007,meyntwee09} and the bound in Equation \ref{eq:BerryEsseen4MCs} hence holds by arguments made in \cite{Kontoyiannis01spectraltheory}.
\end{proof}

An explicit solution to the asymptotic variance of a general irreducible positive recurrent Markov chain is not available. Significant research in the area of steady-state stochastic simulation has focused on obtaining an expression for this quantity \cite{whitt1991, burman1980, KemenySnell1960} and has yielded a closed form solution for the class of homogeneous birth-death processes when $f(\psi_i)$ simply returns the integer valued state itself. 

We build up on an idea introduced in \cite{grassmann1987} to give an explicit closed form solution to the asymptotic variance in Equation \ref{eq:asymVar}.
\begin{theorem}
\label{theorem:asymVariance}
The closed form expressions in Equation \ref{eq:asymVarExp} equals  the asymptotic variance defined in Equation \ref{eq:asymVar}.
\end{theorem}
\begin{proof}
Again we only sketch the proof here and refer to the appendix for a detailed version. For the computation of the sum $\sum_{i=0}^\infty \Cov(f(\Psi_0),f(\Psi_i))$ we will setup and solve a recursion. Grassmann proposed this approach in \cite{grassmann1987} to obtain the asymptotic variance of a continuous time finte state birth death process. 

We define
\begin{align}
r(\psi,i) = \sum_{\psi' \in \mathbb{X}} (f(\psi')-C) \pi_\Psi(\psi') p_{\Psi_{i}| \Psi_0}(\psi|\psi')
\end{align}
Clearly
\begin{align}
r(\psi,0) = (f(\psi)-C)  \pi_\Psi(\psi)
\end{align}
and
\begin{align}
\Cov(f(\Psi_0),f(\Psi_i)) = \sum_{\psi \in \mathbb{X}} (f(\psi)-C) r(\psi,i) =  \sum_{\psi \in \mathbb{X}} f(\psi) r(\psi,i)
\end{align}
Note however that for the computation of the asymptotic variance we actually do not even need to know this covariance for each $i$. It is sufficient to know its sum. So we define 
\begin{align}
R(\psi) = \sum_{i=0}^\infty r(\psi,i)
\end{align}
write
\begin{align}
\label{eq:rec-1}
\sum_{i=0}^\infty \Cov(f(\Psi_0),f(\Psi_i)) = \sum_{\psi \in \mathbb{X}} f(\psi) R(\psi)
\end{align}
and derive an expression for $R(\psi)$.
\end{proof}
\begin{figure*}[h!]
\begin{align}
\label{eq:asymVarExp}
\sigma^2 &=  -\Var(f(\Psi_0)) + 2 \sum_{i=0}^\infty \Cov(f(\Psi_0),f(\Psi_i)) \\
\Var(f(\Psi_0))  &= \log^2(\frac{1}{\bar{\lambda}})\pi_Q(0) + \log^2(\frac{\mu}{\lambda})\mu\overline{\pi_Q(0)}  + \log^2(\frac{{\bar{\mu}}}{\bar{\lambda}}){\bar{\mu}}\overline{\pi_Q(0)} - C^2 \\
c_{M0} &= \frac{\bar{\lambda}}{\bar{\mu}} \left( \mu\log(\frac{\mu}{\lambda})  +  {\bar{\mu}} \log(\frac{{\bar{\mu}}}{\bar{\lambda}}) - C \right) \\
c_{\tilde{M}}  &= \left\{  \frac{c_{M0}}{\mu} + (C - \log\frac{\mu}{\lambda} )  \frac{\pi_Q(0)}{\bar{\mu}} \right\} \\
\label{eq:sumCovExp}
\sum_{i=0}^\infty \Cov(f(\Psi_0),f(\Psi_i)) &= 
\log \frac{1}{\bar{\lambda}}  (-c_{\tilde{M}} \frac{\rho}{ 1 -  \rho } - c_{M0} \rho)
+ \log \frac{\mu}{\lambda} c_{M0} \frac{\rho}{ 1 -  \rho } +
+ \log \frac{\bar{\mu}}{\bar{\lambda}} \frac{\rho}{ 1 -  \rho } (c_{\tilde{M}} - \rho c_{M0})
\end{align}
\end{figure*}


Using the result stated in Theorem \ref{theorem:BerryEsseen4MCs} we can finally prove the core contribution of this paper:
\begin{theorem}
\label{theorem:achievability}
\begin{align}
\log N(n,\epsilon) \geq nC - \sqrt{n} \sigma Q^{-1}(\epsilon) - \frac{1}{2}\log{n} + O(1)
\end{align}
where $C = \pi_\Psi(f)$ and $\sigma$ is defined as in Theorem \ref{theorem:BerryEsseen4MCs}.
\end{theorem}
\begin{proof}
By Theorem \ref{theorem:BerryEsseen4MCs} $\exists A>0:$
\begin{align}
|P((i(\BS{x},\BS{y})-nC)/\sqrt{n \sigma^2} \leq \xi_1) - \Phi(\xi_1)| \leq \frac{A}{\sqrt{n}} \ \forall \xi_1 \in \mathbb{R}
\end{align}
Let $B>A$ and $\xi_1 = \Phi^{-1}(\epsilon - \frac{B}{\sqrt{n}}) < \Phi^{-1}(\epsilon) =  \xi_0$. 
Set $\theta = \sqrt{n} \sigma \xi_1 + nC$ and the application of Lemma \ref{lemma:feinstein} yields
\begin{align}
&\log N(n,\epsilon) - nC - \sqrt{n} \sigma \xi_0 \nonumber \\ 
\geq &\log\left(  \epsilon - P\left(\frac{i(\BS{x},\BS{y})-nC}{\sqrt{n} \sigma } \leq \xi_1\right) \right) + \sqrt{n}\sigma(\xi_1 - \xi_0)\\
\geq &\log\left(  \epsilon -  \Phi(\xi_1) - \frac{A}{\sqrt{n}} \right) + \sqrt{n}\sigma O(\frac{1}{\sqrt{n}}) 
\end{align}
\end{proof}
This confirms that $C$ is the operational capacity of the channel and any rate $R < C$ is achievable. The real beauty of Theorem \ref{theorem:achievability} is, however, that we can use the asymptotic
\begin{align}
\frac{\log N(n,\epsilon)}{n} \sim C - n^{-1/2} \sigma Q^{-1}(\epsilon)
\end{align}
as an approximation to the channel coding rate $\frac{\log N(n,\epsilon)}{n}$ and then anticipate the achievable rate on this channel in the finite block length regime. For illustraton we plotted this asymptotic for blocklengths ranging between $50$ and $3000$ and the example values $\lambda = 0.2$, $\mu = 0.8$ and $\epsilon = 10^{-5}$ in Figure \ref{fig:achievableCodingRate}.
\begin{figure}[!ht]
\psfrag{blocklength n}{blocklength $n$}
\psfrag{achievable coding rate}{achievable coding rate}
\centering
\includegraphics[width = 0.9\columnwidth]{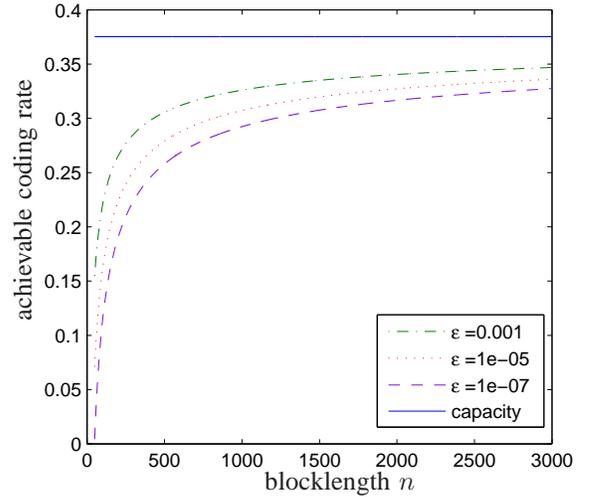}
\caption{Channel Coding Rate in the Finite Block-Length Regime}
\label{fig:achievableCodingRate}
\end{figure}


\section{Proof of Theorem \ref{theorem:burke}}

\section{Proof of Lemma \ref{lemma:feinstein}}
\begin{proof}
Assume $N$ is the maximal size of an $\epsilon$-Code such that $\BS{D}^{(i)} \subset F(\BS{x}^{(i)})$, where
\begin{align}
F(\BS{x}) = \{ \BS{y}: i(\BS{x},\BS{y}) > \theta \}.
\end{align}
Then we have
\begin{align}
P(\BS{D}^{(i)}) = \int_{\BS{D}^{(i)}} P(d\BS{y}) <  \int_{\BS{D}^{(i)}} e^{-\theta} P(d\BS{y}|\BS{x}) \leq e^{-\theta}
\end{align}
and
\begin{align}
P(\cup_i \BS{D}^{(i)}) \leq \sum_i P(\BS{D}^{(i)}) \leq Ne^{-\theta}
\end{align}
Let $\BS{D} = \cup_i \BS{D}^{(i)}$. By the maximality of N it follows that
\begin{align}
P(\BS{D}^c\cap F(\BS{x})|\BS{x}) < 1 - \epsilon.
\end{align}
Or equivalently
\begin{align}
\epsilon < P(\BS{D} \cup F^c(\BS{x})|\BS{x}) \leq P(\BS{D}|\BS{x}) + P(F^c(\BS{x})|\BS{x})
\end{align}
Multiplying this inequality with $P(d\BS{y})$ and integrating it over $\BS{x}$ then yields
\begin{align}
\epsilon \leq P(\BS{D}) + P(i(\BS{x},\BS{y}) \leq \theta)
\end{align}
Putting everything together we obtain the result
\begin{align}
\epsilon - P(i(\BS{x},\BS{y}) \leq \theta) \leq P(\BS{D}) \leq Ne^{-\theta}
\end{align}
\end{proof}

\section{Proof of Theorem \ref{theorem:BerryEsseen4MCs}}
The Markov Chain $\Psi$ is aperiodic and irreducible. The state space of $\Psi_i$ can be chosen to be $\mathbb{X} = \BS{2} \times \mathbb{N} \cup \{ (0,0) \}$. First we verify that Foster's Criterion holds
\begin{lemma}
There exists  a Lyapunov  function $V: \mathbb{X} \to (0,\infty]$, finite at some $\psi_0 \in \mathbb{X}$, a finite set $\mathbb{S} \subset \mathbb{X}$, and $b < \infty$ such that
\begin{align}
\E [ V(\Psi_{i+1}) - V(\Psi_{i})|\Psi_{i} = \psi ] \leq -1 + b \BS{1}_{\mathbb{S}}(\psi), \quad \psi \in \mathbb{X}
\end{align}
Further this function $V$ is Lipschitz, i.e., for some $\alpha > 0$
\begin{align}
|V(y) - V(x)| \leq \alpha ||y-x|| \quad \forall y,x \in \mathbb{X}
\end{align}
and for some $\beta > 0$ and
\begin{align}
\sup_{x \in \mathbb{X}} \E[e^{\beta ||\Psi_{i+1} - \Psi_i ||}| \Psi_i = x] < \infty
\end{align}
\end{lemma}
\begin{proof}
We need to find a function $V$ such that $\E [ V(\Psi_{i+1}) - V(\Psi_{i})|\Psi_{i} = \psi ] \leq -1$ for all but a finite number of $\psi \in \mathbb{X}$. If we simply choose $V(\tilde{y},q) = c q$ for some sufficiently large constant $c>0$ then the requirement is clearly satisfied for all $\psi \in \mathbb{X}$ such that $\tilde{y} = 1$ but it fails to hold otherwise. To fix this shortcoming we reward the transitions to a state with $\tilde{y} = 1$ by a decreasing difference $V(\Psi_{i+1}) - V(\Psi_{i})$. In particular we choose $V(\tilde{y},q) = (q - \tilde{y})/(\mu - \lambda)$. Standard calculations reveal that for that choice $\E [ V(\Psi_{i+1}) - V(\Psi_{i})|\Psi_{i} = \psi ] = -1$ for all $\psi \in \mathbb{X}$ with $ q > 1$. Linear functions are always Lipschitz and  $||\Psi_{i+1} - \Psi_i ||$ is bounded almost surely.
\end{proof}
By the results in \cite{spieksma1994} or Proposition A.5.7. in \cite{CTCNMeyn2007} the chain $\Psi$ is then geometrically ergodic. 

By Theorem A.5.8. in \cite{CTCNMeyn2007} the asymptotic variance $\sigma^2$ is well defined, non-negative and finite, and
\begin{align}
\sigma^2 = \Var(f(\Psi_0)) + 2 \sum_{i=1}^\infty \Cov(f(\Psi_0),f(\Psi_i))
\end{align}
Finally, $f(\tilde{y},q)$ is a bounded, nonlattice (I still have to check that!), real-valued functional on the state space $\mathbb{X}$ and hence
\begin{align}
&P_{\psi_0} \left(\frac{\sum_{i=0}^{n-1} f(\tilde{y}_i,q_i) - n \pi_\Psi(f)}{\sigma \sqrt{n}} \leq \xi \right) - \Phi(\xi) \\
=& \frac{p_\Phi (\xi)}{\sigma \sqrt{n}} \left[ \frac{\eta}{6 \sigma^2}(1-\xi^2) - \hat{f}(\psi_0)  \right] +{o}(n^{-1/2})
\end{align}
where $p_\Phi (\xi)$ denotes the density of the standard Normal distribution $\Phi$, $\hat{f}$ is the solution to Poissons equation and $\eta$ is a constant \cite{Kontoyiannis01spectraltheory}. The solution $\hat{f}$ can be chosen such that $\pi_\Psi(\hat{f}) = 0$ and the claim follows by averaging out $\psi_0$.

\section{Proof of Theorem \ref{theorem:asymVariance}}

Using the representation of $\sigma^2$ in Equation \ref{eq:asymVarianceRep} it remains to find explicit expressions for  $\Var(f(\Psi_0))$ and the sum $\sum_{i=0}^\infty \Cov(f(\Psi_0),f(\Psi_i))$.

The term $\Var(f(\Psi_0))$ is easy to compute
\begin{align}
\Var(f(\Psi_0))  = & \log^2(\frac{1}{\bar{\lambda}})\pi_Q(0) + \log^2(\frac{\mu}{\lambda})\mu\overline{\pi_Q(0)} \nonumber \\
& + \log^2(\frac{{\bar{\mu}}}{\bar{\lambda}}){\bar{\mu}}\overline{\pi_Q(0)} - C^2
\end{align}
Equation \ref{eq:sumCovExp} holds true.
\begin{proof}
For the computation of the sum $\sum_{i=0}^\infty \Cov(f(\Psi_0),f(\Psi_i))$ we will setup and solve a recursion. Grassmann proposed this approach in \cite{grassmann1987} to obtain the asymptotic variance of a continuous time finte state birth death process. 

We define
\begin{align}
r(\psi,i) = \sum_{\psi' \in \mathbb{X}} (f(\psi')-C) \pi_\Psi(\psi') p_{\Psi_{i}| \Psi_0}(\psi|\psi')
\end{align}
Clearly
\begin{align}
r(\psi,0) = (f(\psi)-C)  \pi_\Psi(\psi)
\end{align}
and
\begin{align}
\Cov(f(\Psi_0),f(\Psi_i)) = \sum_{\psi \in \mathbb{X}} (f(\psi)-C) r(\psi,i) =  \sum_{\psi \in \mathbb{X}} f(\psi) r(\psi,i)
\end{align}
Note however that for the computation of the asymptotic variance we actually do not even need to know this covariance for each $i$. It is sufficient to know its sum. So we define 
\begin{align}
R(\psi) = \sum_{i=0}^\infty r(\psi,i)
\end{align}
write
\begin{align}
\label{eq:rec-1}
\sum_{i=0}^\infty \Cov(f(\Psi_0),f(\Psi_i)) = \sum_{\psi \in \mathbb{X}} f(\psi) R(\psi)
\end{align}
and derive a recursion for $R(\psi)$.

For mean ergodic Markov processes $p_{\Psi_{i}| \Psi_0}(\psi|\psi') \to \pi_\Psi(\psi)$ as $i \to \infty$  and hence 
\begin{align}
\lim_{i \to \infty} r(\psi,i) = 0
\end{align}
Summing $r(\psi,i+1) - r(\psi,i)$ in $i$ from zero to infinity then clearly yields $(C-f(\psi))  \pi_\Psi(\psi)$.
By the Chapman-Kolmogorov equations
\begin{align}
&p_{\Psi_{i+1}| \Psi_0}(\psi|\psi') - p_{\Psi_{i}| \Psi_0}(\psi|\psi') \nonumber \\
&= \sum_{\psi'' \in \mathbb{X}} p_{\Psi_{i}| \Psi_0}(\psi''|\psi') 
\left\{p_{\Psi_{i+1}| \Psi_i}(\psi|\psi'')- \delta_{\psi,\psi''}\right\}
\end{align}
and thus
\begin{align}
&r(\psi,i+1)-r(\psi,i) \nonumber \\
&=\sum_{\psi' \in \mathbb{X}} (f(\psi')-C) \pi_\Psi(\psi') \left\{ p_{\Psi_{i+1}| \Psi_0}(\psi|\psi') - p_{\Psi_{i}| \Psi_0}(\psi|\psi') \right\} \nonumber \\
&= \sum_{\psi'' \in \mathbb{X}} \left\{ p_{\Psi_{i+1}| \Psi_i}(\psi|\psi'')- \delta_{\psi,\psi''} \right\}
r(\psi'',i)
\end{align}
If this expression for  $r(\psi,i+1) - r(\psi,i)$ is also summed in $i$ from zero to infinity and then compared to the above result of the same sum, we obtain
\begin{align}
\label{eq:rec0}
(C-f(\psi))  \pi_\Psi(\psi) = \sum_{\psi'' \in \mathbb{X}} \left\{ p_{\Psi_{i+1}| \Psi_i}(\psi|\psi'')- \delta_{\psi,\psi''} \right\} R(\psi'')
\end{align}
For notational convenience we abbreviate the right hand-side of Equation \ref{eq:rec0} by $D(\psi)$.

For $\psi$ with $q>0$ and $\tilde{y} = 0$
\begin{align}
\label{eq:rec1}
D(\psi) =&(\bar{\lambda}\bar{\mu}-1)R(q,0) + \bar{\lambda}\bar{\mu}R(q+1,1) + \lambda\bar{\mu}R(q,1) \nonumber \\ &+ \lambda\bar{\mu}R(q-1,0)
\end{align}
For $\psi$ with $q>0$ and $\tilde{y} = 1$
\begin{align}
\label{eq:rec2}
D(\psi) =& ({\lambda}{\mu}-1)R(q,1) + \bar{\lambda}{\mu}R(q+1,1) + \bar{\lambda}{\mu}R(q,0) \nonumber \\ &+ \lambda{\mu}R(q-1,0)
\end{align}
And for $\psi$ with $q=0$ and $\tilde{y} = 0$
\begin{align}
\label{eq:rec3}
D(\psi) = -\lambda R(0,0) + \bar{\lambda}R(1,1) 
\end{align}
Adding Equations \ref{eq:rec1} and \ref{eq:rec2} yields
\begin{align}
\label{eq:rec4}
\bar{\lambda} R(q+1,1) - \lambda R(q,0) - \bar{\lambda}R(q,1) + \lambda R(q-1,0)
\end{align}
We now sum Equation \ref{eq:rec0} in two ways:
\begin{align}
\label{eq:rec5}
\sum_{\psi' \in \mathbb{X}: q' \leq q} D(\psi) = \bar{\lambda} R(q+1,1) - \lambda R(q,0) = M(q,0)
\end{align}
for $q \geq 0$ where we defined
\begin{align}
M(q,0) = \sum_{\psi' \in \mathbb{X}: q' \leq q} (C-f(\psi))  \pi_\Psi(\psi)
\end{align}
and
\begin{align}
\label{eq:rec6}
\sum_{\psi' \in \mathbb{X}: q' < q || q'=q,\tilde{y}=1} D(\psi) &= -\lambda \bar{\mu} R(q-1,0) -\lambda \bar{\mu} R(q,1) \nonumber \\ 
+& \bar{\lambda}\mu R(q,0) + \bar{\lambda}\mu R(q+1,1) = M(q,1)
\end{align}
for $q \geq 1$ where we defined 
\begin{align}
M(q,1) = \sum_{\psi \in \mathbb{X}: q' < q || q'=q,\tilde{y}=1} (C-f(\psi))  \pi_\Psi(\psi)
\end{align}
We can combine Equation \ref{eq:rec5} and \ref{eq:rec6} to obtain the first order recurrence
\begin{align}
\label{eq:rec7}
R(q+1,0) = \frac{\lambda \bar{\mu}}{\bar{\lambda}\mu} R(q,0) + \tilde{M}(q)
\end{align}
for $q \geq 0$ where
\begin{align}
\label{eq:rec8}
\tilde{M}(q) = \frac{1}{\mu}M(q+1,1)- M(q+1,0) + \frac{\lambda \bar{\mu}}{\bar{\lambda}\mu} M(q,0)
\end{align}
Note that
\begin{align}
\label{eq:rec9}
M(q,0) &= C \left( 1 - \frac{\pi_Q(0)}{\bar{\mu}}  \frac{\left( \frac{\lambda \bar{\mu}}{\bar{\lambda}\mu} \right)^{q+1}  }{1- \frac{\lambda \bar{\mu}}{\bar{\lambda}\mu}} \right) - \sum_{\psi \in \mathbb{X}: q' \leq q} f(\psi)  \pi_\Psi(\psi) \nonumber \\
&=  \frac{\pi_Q(0)}{\bar{\mu}}\frac{\left( \frac{\lambda \bar{\mu}}{\bar{\lambda}\mu} \right)^{q+1}  }{1- \frac{\lambda \bar{\mu}}{\bar{\lambda}\mu}} \left( \mu\log(\frac{\mu}{\lambda})  +  {\bar{\mu}} \log(\frac{{\bar{\mu}}}{\bar{\lambda}}) - C \right) \nonumber \\
&=  \frac{\bar{\lambda}}{\bar{\mu}} \left( \frac{\lambda \bar{\mu}}{\bar{\lambda}\mu} \right)^{q+1}  \left( \mu\log(\frac{\mu}{\lambda})  +  {\bar{\mu}} \log(\frac{{\bar{\mu}}}{\bar{\lambda}}) - C \right) \nonumber \\
&= c_{M0} \left( \frac{\lambda \bar{\mu}}{\bar{\lambda}\mu} \right)^{q+1}
\end{align}
and with this Equation \ref{eq:rec8} becomes
\begin{align}
\tilde{M}(q) = \frac{1}{\mu} M(q+1,1)
\end{align}
But
\begin{align}
M(q+1,1) &= M(q,0) + (C - \log\frac{\mu}{\lambda} ) \pi_Q(0) \frac{\mu}{\bar{\mu}} \left( \frac{\lambda \bar{\mu}}{\bar{\lambda}\mu} \right)^{q+1} \nonumber \\
&= \left( \frac{\lambda \bar{\mu}}{\bar{\lambda}\mu} \right)^{q+1} \left\{ c_{M0} + (C - \log\frac{\mu}{\lambda} ) \pi_Q(0) \frac{\mu}{\bar{\mu}} \right\} \nonumber
\end{align}
So we obtain
\begin{align}
\label{eq:rec11}
\tilde{M}(q) &= \left( \frac{\lambda \bar{\mu}}{\bar{\lambda}\mu} \right)^{q+1} \left\{  \frac{c_{M0}}{\mu} + (C - \log\frac{\mu}{\lambda} )  \frac{\pi_Q(0)}{\bar{\mu}} \right\} \nonumber \\
&= c_{\tilde{M}} \left( \frac{\lambda \bar{\mu}}{\bar{\lambda} \mu} \right)^{q+1}
\end{align}
and the generating function
\begin{align}
\tilde{M}(z) = \sum_{q \geq 0} \tilde{M}(q) z^q = \frac{c_{\tilde{M}} \frac{\lambda \bar{\mu}}{\bar{\lambda} \mu} }{1 - z \frac{\lambda \bar{\mu}}{\bar{\lambda} \mu}}
\end{align}

We now define two new sequences $a_q$ and $b_q$ such that
\begin{align}
\label{eq:rec12}
R(q,0) = a_q + b_q R(0,0)
\end{align}
Clearly, $a_0 = 0$ and $b_0 = 1$. By substituting Equation \ref{eq:rec12} into Equation \ref{eq:rec7} we find that
\begin{align}
a_{q+1} =  \frac{\lambda \bar{\mu}}{\bar{\lambda}\mu} a_q + \tilde{M}(q)
\end{align}
and
\begin{align}
b_{q+1} = \frac{\lambda \bar{\mu}}{\bar{\lambda}\mu} b_q
\end{align}
The solution to the recurrence $b_q$ is obvious
\begin{align}
b_q = \left( \frac{\lambda \bar{\mu}}{\bar{\lambda}\mu} \right)^q
\end{align}
In order to obtain the solution to the recurrence $b_q$ we employ the generating function method \cite{west}
\begin{align}
\sum_{q \geq 0} a_{q+1} z^q = \frac{\lambda \bar{\mu}}{\bar{\lambda}\mu} \sum_{q \geq 0} a_q z^q + \sum_{q \geq 0} \tilde{M}(q) z^q
\end{align}
and therefore
\begin{align}
z^{-1} A(z) = \frac{\lambda \bar{\mu}}{\bar{\lambda}\mu}  A(z) + \tilde{M}(z)
\end{align}
We can now solve this equation for $A(z)$ to obtain
\begin{align}
A(z) = \frac{\tilde{M}(z)}{z^{-1} -  \frac{\lambda \bar{\mu}}{\bar{\lambda}\mu}} = c_{\tilde{M}} \frac{\frac{\lambda \bar{\mu}}{\bar{\lambda} \mu}z}{\left( 1 -  \frac{\lambda \bar{\mu}}{\bar{\lambda}\mu}z \right)^2}
\end{align}
and the corresponding sequence
\begin{align}
a_q = c_{\tilde{M}} q \left( \frac{\lambda \bar{\mu}}{\bar{\lambda}\mu} \right)^q
\end{align}
Finally
\begin{align}
\label{eq:rec13}
R(q,0) = a_q + b_q R(0,0) 
= c_{\tilde{M}} q \left( \frac{\lambda \bar{\mu}}{\bar{\lambda}\mu} \right)^q
+ \left( \frac{\lambda \bar{\mu}}{\bar{\lambda}\mu} \right)^q R(0,0)
\end{align}

Further
\begin{align}
\sum_{\psi \in \mathbb{X}} r(\psi,i) = \sum_{\psi' \in \mathbb{X}} (f(\psi')-C) \pi_\Psi(\psi') \sum_{\psi \in \mathbb{X}}  p_{\Psi_{i}| \Psi_0}(\psi|\psi') = 0
\end{align} and summing this equation in $i$ yields
\begin{align}
\label{eq:rec14}
\sum_{\psi \in \mathbb{X}} R(\psi) = 0
\end{align}
This result now allows us to compute $R(0,0)$: Using Equation \ref{eq:rec5} we can write
\begin{align}
\label{eq:rec15}
\sum_{\psi \in \mathbb{X}} R(\psi) &= \sum_{q \geq 0} R(q,0) + \sum_{q \geq 0} \left( \frac{\lambda}{\bar{\lambda}} R(q,0) + \frac{1}{\bar{\lambda}}M(q,0) \right) \nonumber \\
&= \frac{1}{\bar{\lambda}} \sum_{q \geq 0} \left( R(q,0) + M(q,0) \right)
\end{align}
Combining Equations \ref{eq:rec12},\ref{eq:rec14} and \ref{eq:rec15} then yields
\begin{align}
R(0,0) &= - \frac{\sum_{q \geq 0} a_q + \sum_{q \geq 0} M(q,0) }{\sum_{q \geq 0} b_q } \nonumber \\
&= -c_{\tilde{M}} \frac{\frac{\lambda \bar{\mu}}{\bar{\lambda} \mu}}{ 1 -  \frac{\lambda \bar{\mu}}{\bar{\lambda}\mu} } - c_{M0} \frac{\lambda \bar{\mu}}{\bar{\lambda} \mu}
\end{align}

Using the expressions for $R(j+1,1)$, $R(j,0)$ and $M(j,0)$ from Equations \ref{eq:rec5}, \ref{eq:rec13} and \ref{eq:rec9}, respectively, we are eventually in a position to simplify the expression  in Equation \ref{eq:rec-1}:
\begin{align}
&\sum_{i=0}^\infty \Cov(f(\Psi_0),f(\Psi_i)) =
\log \frac{1}{\bar{\lambda}} R(0,0) \nonumber \\
&+  \log \frac{\mu}{\lambda}
\left( \frac{\lambda}{\bar{\lambda}} \sum_{q \geq 0} R(q,0) + \frac{1}{\bar{\lambda}} \sum_{q \geq 0} M(q,0) \right)
+   \log \frac{{\bar{\mu}}}{\bar{\lambda}} \sum_{q > 0} R(q,0)
\end{align}
where
\begin{align}
\sum_{q \geq 0} R(q,0) &= c_{\tilde{M}} \frac{\frac{\lambda \bar{\mu}}{\bar{\lambda} \mu}}{\left( 1 -  \frac{\lambda \bar{\mu}}{\bar{\lambda}\mu} \right)^2 } + \frac{R(0,0)}{1 - \frac{\lambda \bar{\mu}}{\bar{\lambda}\mu}} \nonumber \\
&= - c_{M0} \frac{\frac{\lambda \bar{\mu}}{\bar{\lambda}\mu}}{1 - \frac{\lambda \bar{\mu}}{\bar{\lambda}\mu}}
\end{align}
and
\begin{align}
\sum_{q \geq 0} M(q,0) = c_{M0} \frac{\frac{\lambda \bar{\mu}}{\bar{\lambda}\mu}}{1 - \frac{\lambda \bar{\mu}}{\bar{\lambda}\mu}}
\end{align}
The claim then follows.
\end{proof}

\bibliography{IEEEabrv,biblio} 
\bibliographystyle{IEEEtran}

\end{document}